\documentclass[reprint, prd]{revtex4-1}

\usepackage{CJK, hyperref, amsmath, amsthm}

\newtheorem{lemma}{Lemma}
\newtheorem{theorem}{Theorem}

\theoremstyle{definition}
\newtheorem{definition}{Definition}
\newtheorem{example}{Example}

\theoremstyle{remark}
\newtheorem*{remark}{Remark}

\DeclareMathOperator{\tr}{tr}

\begin{document}

\begin{CJK*}{UTF8}{}

\title{Eigenstate entanglement in the Sachdev-Ye-Kitaev model}

\CJKfamily{gbsn}

\author{Yichen Huang (黄溢辰)}
\email{yichen.huang@microsoft.com}
\affiliation{Institute for Quantum Information and Matter, California Institute of Technology, Pasadena, California 91125, USA}

\author{Yingfei Gu (顾颖飞)}
\email{yingfei\_gu@g.harvard.edu}
\affiliation{Department of Physics, Harvard University, Cambridge, Massachusetts 02138, USA}

\date{\today}

\begin{abstract}

We study the entanglement entropy of eigenstates (including the ground state) of the Sachdev-Ye-Kitaev model. We argue for a volume law, whose coefficient can be calculated analytically from the density of states. The coefficient depends on not only the energy density of the eigenstate but also the subsystem size. Very recent numerical results of Liu, Chen, and Balents confirm our analytical results.

\end{abstract}

\maketitle

\end{CJK*}

\section{Introduction}

Entanglement, a concept of quantum information theory, has been widely used in condensed matter and high energy physics \cite{BKLS86, Sre93, RT06a, RT06b} to provide insights beyond those obtained via ``conventional'' quantities. For ground states of (geometrically) local Hamiltonians, it characterizes quantum criticality \cite{HLW94, VLRK03, LRV04, CC04, CC09} and topological order \cite{KP06, LW06, LH08}. The scaling of entanglement \cite{ECP10} is quantitatively related to the classical simulability of quantum many-body systems \cite{VC06, SWVC08, Osb12, GHLS15, Hua15}.

Besides ground states, it is also important to understand the entanglement of excited eigenstates. Significant progress has been made in this regard \cite{Deu10, SPR12, DLS13, HNO+13, BN13, HM14, KLW15, YCHM15, DKPR16, VHBR17, VR17, DLL18, NWFS+18, GG15, Hua19, LG17}. For chaotic local Hamiltonians, one expects a volume law: The entanglement entropy of an eigenstate between a subsystem (smaller than half the system size) and its complement scales as the subsystem size with a coefficient depending on the energy density of the eigenstate. Indeed, analytical arguments \cite{Deu10} and numerical simulations \cite{SPR12, DLS13, GG15} strongly suggest that it is to leading order equal to the thermodynamic entropy of the subsystem at the same energy density.

Let us put this result in context. The eigenstate thermalization hypothesis (ETH) states that for expectation values of local observables, a single eigenstate resembles a thermal state with the same energy density \cite{Deu91, Sre94, RDO08}. The equivalence of entanglement and thermodynamic entropies is a variant of ETH for the von Neumann entropy.

Over the past few years, the Sachdev-Ye-Kitaev (SYK) model \cite{SY93, Kit15, MS16} has become a very active research topic in condensed matter and high energy physics. As a quantum mechanical model of $N\gg1$ Majorana fermions with random all-to-all interactions, it is exactly solvable in the large-$N$ limit, and has an extensive zero-temperature entropy. Entanglement and ETH in the SYK model have recently been studied \cite{FS16, gu2017spread, kourkoulou2017pure, SV17, HLZ18}. In particular, the numerical results of Fu and Sachdev \cite{FS16} suggest the breakdown of the equivalence of entanglement and thermodynamic entropies in the SYK model: The ground-state entanglement entropy appears to be slightly less than the maximum ($\frac{1}{2}\ln2$ per Majorana fermion) and is significantly greater than the zero-temperature entropy.

We argue that due to all-to-all interactions, the equivalence of entanglement and thermodynamic entropies should be modified as follows: The entanglement entropy of an eigenstate (including the ground state) for a subsystem smaller than half the system size is to leading order equal to the thermodynamic entropy of the subsystem at a different energy density, which depends on not only the energy density of the eigenstate but also the subsystem size. This allows us to derive an analytical expression for the scaling of the eigenstate entanglement entropy in the SYK model from the density of states.

The SYK model is maximally chaotic \cite{Kit15, MS16} in the sense of the maximum ``Lyapunov exponent'' \cite{MSS16} for the exponential growth of out-of-time-ordered correlators \cite{LO69, SS14b, RSS15, SS15, Kit14}. Our argument is not related to this dynamical behavior, and should apply to a broad class of chaotic quantum many-body systems.

\section{Preliminaries}

\begin{definition} [entanglement entropy]
The entanglement entropy of a bipartite pure state $\rho_{AB}$ is defined as the von Neumann entropy
\begin{equation}
S(\rho_A)=-\tr(\rho_A\ln\rho_A)
\end{equation}
of the reduced density matrix $\rho_A=\tr_B\rho_{AB}$.
\end{definition}

Consider a quantum mechanical system of $N$ Majorana fermions $\chi_1,\chi_2,\ldots,\chi_N$, where $N\gg1$ is an even number. The Hamiltonian of the SYK model is
\begin{equation} \label{syk}
H=\sum_{1\le i<j<k<l\le N}J_{ijkl}\chi_i\chi_j\chi_k\chi_l,\quad\{\chi_i,\chi_j\}=\delta_{ij},
\end{equation}
where the coefficients are independent real Gaussian random variables with zero mean $\overline{J_{ijkl}}=0$ and variance
\begin{equation} \label{normal}
\overline{J_{ijkl}^2}= \frac{3!}{N^3}.
\end{equation}

We summarize the spectral properties of the SYK model in the thermodynamic limit $N\to+\infty$:
\begin{itemize}
\item The mean energy of $H$ is $\tr H/d=0$, where $d=2^{N/2}$ is the dimension of the Hilbert space.
\item The distribution of eigenvalues near the mean energy is asymptotically normal with variance $\tr(H^2)/d=\Theta(N)$. Asymptotic normality is a general feature of models with random few-body interactions \cite{ES14}.
\item The ground-state energy is $-NE_0+o(N)$ (extensive), where $E_0>0$ is a known constant.
\item The density of states at energy $NE$ is \cite{GV17}
\begin{equation}
D_N(E)\sim e^{\left[\frac{1}{2}\ln 2-\frac{1}{16}\arcsin^2\left(\frac{E}{E_0} \right)\right]N}.
\end{equation}
This is an excellent approximation away from the edges of the spectrum, i.e., when $E_0-|E|$ is not too small.
\end{itemize}

\section{Argument}

We divide the system into two subsystems $A$ and $B$. Subsystem $A$ consists of $M$ Majorana fermions, where $M$ is even. Assume without loss of generality that $M\le N/2$. We split the Hamiltonian (\ref{syk}) into three parts:
\begin{equation}
H=H_A+H_\partial+H_B,
\end{equation}
where $H_{A(B)}$ contains terms acting only on subsystem $A(B)$, and $H_\partial$ consists of cross terms. We observe that
\begin{equation} \label{r}
H'_A:=\left(\frac{N}{M}\right)^\frac{3}{2}H_A
\end{equation}
is the SYK model of $M$ Majorana fermions.

\begin{lemma} [\cite{Weh78}] \label{l1}
The thermal state maximizes the von Neumann entropy among all states with the same energy.
\end{lemma}

\begin{proof}
For completeness, we give a simple proof of this well-known fact. Let
\begin{equation}
\sigma_\beta:=\frac{1}{Z_\beta}e^{-\beta H},\quad Z_\beta:=\tr e^{-\beta H}
\end{equation}
be a thermal state at inverse temperature $\beta$. For any density matrix $\rho$, the relative entropy
\begin{equation}
S(\rho\|\sigma_\beta):=\tr(\rho\ln\rho-\rho\ln\sigma_\beta)\ge0
\end{equation}
is nonnegative. If $\rho$ and $\sigma_\beta$ have the same energy, then
\begin{multline}
S(\rho)=-\tr(\rho\ln\rho)\le-\tr(\rho\ln\sigma_\beta)=\beta\tr(\rho H)+\ln Z_\beta\\
=\beta\tr(\sigma_\beta H)+\ln Z_\beta=-\tr(\sigma_\beta\ln\sigma_\beta)=S(\sigma_\beta).
\end{multline}
This completes the proof of Lemma \ref{l1}.
\end{proof}

\begin{lemma} \label{l2}
To leading order,
\begin{equation}
S(\rho_A)\le\ln D_M\left(\frac{\tr(\rho_AH'_A)}{M}\right)
\end{equation}
for any density matrix $\rho_A$ of subsystem $A$.
\end{lemma}

\begin{proof}
Consider the SYK model $H'_A$. The inverse temperature $\beta$ at which the thermal state $\sigma_\beta:=e^{-\beta H'_A}/\tr e^{-\beta H'_A}$ has the same energy as $\rho_A$ is obtained by solving
\begin{equation}
\frac{\int_{-E_0}^{E_0}M\epsilon e^{-\beta M\epsilon}D_M(\epsilon)\mathrm d\epsilon}{\int_{-E_0}^{E_0}e^{-\beta M\epsilon}D_M(\epsilon)\mathrm d\epsilon}=\tr(\rho_AH'_A).
\end{equation}
The saddle-point approximation, which becomes exact in the limit $M\to+\infty$, gives
\begin{equation} \label{s}
\beta=\frac{\partial\ln D_M(\epsilon)}{M\partial\epsilon}\Big|_{\epsilon=\frac{\tr(\rho_AH'_A)}{M}}.
\end{equation}
Using Lemma \ref{l1},
\begin{multline}
S(\rho_A)\le S(\sigma_\beta)=\ln\int_{-E_0}^{E_0}e^{\beta(\tr(\rho_AH'_A)-M\epsilon)}D_M(\epsilon)\mathrm d\epsilon\\
=\ln D_M\left(\frac{\tr(\rho_AH'_A)}{M}\right).
\end{multline}
In the last step, we used the saddle-point approximation and Eq. (\ref{s}).
\end{proof}

\begin{theorem} \label{t1}
Let $|\psi\rangle$ be an eigenstate of the SYK model (\ref{syk}) with energy $NE$. To leading order,
\begin{equation} \label{up}
S(\rho_A)\le\ln D_M\left(\left(\frac{M}{N}\right)^{\frac{3}{2}}E\right),
\end{equation}
where $\rho_A=\tr_B|\psi\rangle\langle\psi|$ is the reduced density matrix. 
\end{theorem}

\begin{proof}
To leading order, the quartic Hamiltonian (\ref{syk}) has $N^4/4!$ terms, $M^4/4!$ of which are in $H_A$. In average,
\begin{equation}
\tr(\rho_AH_A)=\left(\frac{M}{N}\right)^4NE.
\end{equation}
Substituting Eq. (\ref{r}),
\begin{equation}
\tr(\rho_AH'_A)=\left(\frac{M}{N}\right)^\frac{5}{2}NE=\left(\frac{M}{N}\right)^\frac{3}{2}ME.
\end{equation}
Hence, Theorem \ref{t1} follows from Lemma \ref{l2}.
\end{proof}

A slight modification of the proof(s) of Theorem \ref{t1} (and Lemma \ref{l2}) yields

\begin{theorem} \label{t2}
For (translationally invariant) local Hamiltonians, the entanglement entropy of an eigenstate between two subsystems is to leading order upper bounded by the thermodynamic entropy of the smaller subsystem at the same energy density.
\end{theorem}

\begin{remark}
The upper bound in Theorem \ref{t2} holds regardless of whether the system is chaotic. Notably, it is not tight in (integrable) free-fermion systems \cite{VHBR17}. In chaotic systems, analytical arguments \cite{Deu10} and numerical simulations \cite{SPR12, DLS13, GG15} strongly suggest that the bound is attained.
\end{remark}

Since the SYK model is chaotic, one might expect that the upper bound in Theorem \ref{t1} is attained. Thus, we have a volume law
\begin{equation}
S(\rho_A)\sim\left[\frac{1}{2}\ln 2-\frac{1}{16}\arcsin^2\left(\left( \frac{M}{N}\right)^{\frac{3}{2}}\frac{E}{E_0}\right)\right]M
\end{equation}
with a coefficient depending on $E$ (the energy density of the eigenstate) and $M$ (the size of the smaller subsystem). This is the leading-order scaling of the eigenstate entanglement entropy, and we are unable to calculate the subleading corrections.

\begin{example}
Suppose that $N$ is a multiple of $4$. For $M=N/2$, the ground-state entanglement entropy scales as
\begin{multline} \label{ana}
S(\rho_A)\sim\left(\frac{1}{2}\ln 2-\frac{1}{16}\arcsin^2\frac{1}{2\sqrt2}\right)M\\
\approx(0.34657-0.00816)M=0.33841M.
\end{multline}
\end{example}

\section{Conclusions and outlook}

In the SYK model, we have argued that the entanglement entropy of an eigenstate with energy $NE$ for a subsystem of size $M\le N/2$ is to leading order equal to the thermodynamic entropy of the subsystem at energy $(M/N)^{3/2}ME$. Therefore,
\begin{itemize}
\item The entanglement entropy of an eigenstate obeys a volume law with the maximum coefficient $\frac{1}{2}\ln2$ if the subsystem size is a vanishing fraction of the system size. This is because the subsystem is at the mean energy density of the Hamiltonian (\ref{r}).
\item The entanglement entropy of an eigenstate with finite energy density obeys a volume law with a non-maximal coefficient if the subsystem size is a constant fraction of the system size.
\end{itemize}

In the future, it would be interesting to study the Renyi entanglement entropy of eigenstates of the SYK model. As a generalization of entanglement entropy, the Renyi entanglement entropy reflects the entanglement spectrum (the full spectrum of the reduced density matrix $\rho_A$). See Refs. \cite{GG15, DLL18, LG17} for recent results on the Renyi entanglement entropy of eigenstates of chaotic local Hamiltonians.

\begin{acknowledgments}

We would like to thank Xiao-Liang Qi for pointing out Lemma \ref{l1}. Y.H. acknowledges funding provided by the Institute for Quantum Information and Matter, an NSF Physics Frontiers Center (NSF Grant PHY-1733907). Y.H. is also supported by NSF DMR-1654340. Y.G. is supported by the Gordon and Betty Moore Foundation EPiQS Initiative through Grant GBMF-4306.

\end{acknowledgments}

\emph{Note added.}|Very recently, we became aware of related work by Liu \emph{et al.} \cite{LCB17}, which studied eigenstate entanglement in the SYK model using different methods. Among other results, the ground-state entanglement entropy was calculated up to $N=44$ Majorana fermions using exact diagonalization. For $M=N/2$ (the subsystem size is half the system size), the data are well fitted by the expression $0.3375M-0.666$, which is consistent with our analytical result (\ref{ana}).

\bibliographystyle{abbrv}
\bibliography{syk}

\end{document}